\documentclass[11pt]{article}             
\usepackage{osid}                         %

\usepackage{amsmath,amssymb,amscd,amsthm,mathrsfs}
\usepackage{bbm}
\usepackage{comment} 
\usepackage{enumerate} 
\usepackage{color}

\usepackage[bookmarks=false,pdfstartview={FitH}]{hyperref}

\theoremstyle{plain}
\newtheorem{lemma}{Lemma}
\newtheorem{thm}{Theorem}

\newtheorem{proposition}{Proposition}

\theoremstyle{definition}
\newtheorem{defi}{Definition}

\theoremstyle{remark}
\newtheorem{remark}{Remark}

\title{Finite-Dimensional Stinespring Curves Can Approximate Any Dynamics} 
\author{Frederik vom Ende
	\\[1mm]{\footnotesize\it Dahlem Center for Complex Quantum
Systems, Freie Universität Berlin, Arnimallee 14, 14195 Berlin, Germany
 \& {frederik.vom.ende@fu-berlin.de}}\\[2ex]
}

\AtBeginDocument{
  \label{CorrectFirstPageLabel}
  
}

\begin{document}

\maketitle
\begin{abstract}
We generalize a recent result stating that all analytic quantum dynamics can be represented exactly as the reduction of unitary dynamics generated by a time-dependent Hamiltonian. More precisely, we prove that the partial trace over analytic paths of unitaries can approximate any Lipschitz-continuous quantum dynamics arbitrarily well. Equivalently, all such dynamics can be approximated by analytic Kraus operators. We conclude by discussing potential improvements and generalizations of these results, their limitations, and the general challenges one has to overcome when trying to relate dynamics to quantities on the system-environment level.
\end{abstract}

\section{Introduction}
Combining fundamental concepts from quantum information---such as Kraus operators or Stinespring dilations---with quantum dynamics is an area of research which, surprisingly, only gained traction within the last decade or two.
Examples of this are recent investigations \cite{Burgarth22,vE22_Stinespring} into an old result of Davies on (infinite-dimensional) unitary dilations of (finite-dimensional) dynamical semigroups \cite{Davies72,Davies76}, as well as works on expressing system dynamics via time-dependent Kraus operators \cite{ACH07,HCLA14,ABCM14}.
Such sets of Kraus operators which are continuous w.r.t.~some parameter (not necessarily time) could also be useful for optimization tasks beyond parametrized unitary circuits \cite{OGB21}.
Moreover, the existence of a continuous curve of Stinespring unitaries for time-independent Markovian dynamics \cite{Dive15} has recently been used for studying quantum-embeddability of stochastic matrices \cite{Shahbeghi23}.

At first glance, the task of constructing a continuous curve of Stinespring unitaries may appear trivial: given any dynamics $\Phi(t)$,
for all $t$ one can choose a unitary $U_t$
such that $\Phi(t)={\rm tr}_{E}(U_t((\cdot)\otimes|0\rangle\langle 0|)U_t^*)$
by the ``static'' version of 
Stinespring's representation theorem \cite[Thm.~6.18]{Holevo12}.
The caveat is that, in general, the resulting map
$U:t\mapsto U_t$ 
will lack any kind of structure, because we picked the unitaries independent of each other.
In particular, there is no reason for $U$ to be the solution to
Schr\"odinger's equation,
meaning it does not describe the dynamics of a larger closed system.
We emphasize that this issue has nothing to do with the representation we picked: the same problem arises on the level of Stinespring isometries.
While there, the only degree of freedom is a unitary on the environment \cite[Coro.~2.24]{Watrous18},
despite a continuity result for Stinespring isometries \cite{Kretschmann08,Kretschmann08_2}
it is unclear how to make an arbitrary
curve of isometries $t\mapsto V(t)$ (that gives rise to continuous dynamics)
continuous via environment unitaries,
i.e.~it is not clear how to construct $t\mapsto U(t)$ unitary such that $t\mapsto (\mathbbm1\otimes U(t))V(t)$ is continuous.
The fact that this problem is independent of the chosen formalism is also in line (i) with
the fact that the Kraus and the Stinespring representation are equivalent, even for dynamics (Lemma~\ref{lemma_equiv_repr} below)
and (ii) with a remark of Pollock and Modi from \cite{PM18} that
``recovering the underlying [system-environment] quantities [...] from operationally reconstructible time-dependent maps [...] is neither uniquely constrained nor easy to achieve in practice.'' 

On the other hand, the heuristic derivation of Markovian dynamics is fundamentally tied to Stinespring dilations, as the former
are often introduced as the reduced dynamics of a larger closed system---assuming the system in question satisfies usual approximations (Born-Markov, separation of timescales, etc.)
\cite{BreuPetr02}.
Yet, this connection was made rigorous only about a decade ago: as mentioned above, Dive et al.~\cite{Dive15} showed that every quantum-dynamical semigroup can be written as the reduction of (finite-dimensional, time-dependent) unitary dynamics which are analytic everywhere aside from time zero;
for the precise statement we refer to Prop.~\ref{prop_dive} below.
The key step in their proof is to diagonalize the path of Choi matrices to obtain dynamic Kraus operators.
For this they used an old result of Kato \cite[Ch.~2 §6.2]{Kato80} on diagonalizing analytic paths of Hermitian matrices.

It may come as a surprise that the assumption of the dynamics being analytic is necessary for this proof strategy to work:
an example of a smooth, non-analytic curve of Hermitian matrices which does not admit any path of eigenvectors that is continuous at time zero is given in \cite[Ch.~2, Ex.~5.3]{Kato80}.
And this is not the only issue that comes up: Dive et al.~also observed
that the unitaries not being analytic at time zero is due to divergence of the generating Hamiltonian (cf.~also \cite[Coro.~4.5]{Burgarth22}).
This is not a flaw of the construction, but rather an unavoidable phenomenon:
If the unitary dynamics were differentiable at zero, then the derivative of the reduced dynamics has to be tangent to the set of quantum channels, i.e.~at time zero the reduced dynamics ``look like'' unitary dynamics \cite{vE22_Stinespring}.
But this contradicts the dissipative nature of (non-unitary) Markovian dynamics.

Altogether, the problem of fully connecting quantum dynamics with quantities 
on the system-environment level is still open.
In this paper we take a first step in this direction by looking at a relaxed version 
of this problem: We show that all dynamics can at least be approximated 
arbitrarily well via the partial trace over the dynamics of a larger (closed) system, and the 
Hamiltonian of the latter dynamics can always be chosen uniformly bounded, and even analytic for finite times
(Thm.~\ref{thm0} below).
This is also motivated by another observation made by Dive et al.~\cite{Dive15}:
while the 
overall
Hamiltonian 
has to diverge at zero, approximating it by something bounded (in time) only results in a ``small'' error on the level of the reduced dynamics.
\section{Preliminaries: Quantum Dynamics \& Stinespring Curves}\label{sec:prelim}
First some notation to set the stage: 
the collection of all linear maps on $n\times n$-matrices will be denoted by $\mathcal L(\mathbb C^{n\times n})$ while the subset of completely positive, trace-preserving maps (also called quantum maps or quantum channels) will be called\footnote{
While we work in the Schr\"odinger picture, all of our results have analogous formulations in the Heisenberg picture by means of the usual duality \cite{vE_dirr_semigroups}.
}
$\mathsf{CPTP}(n)$.
We will write 
$\mathbb D(\mathbb C^n)$ 
for the set of all $n$-level quantum states
which, as usual, is equipped with the trace norm $\|\cdot\|_1$ (i.e.~the sum of all singular values of the input).
In contrast $\|\cdot\|_\infty$ refers to the operator norm (on matrices) which is given by the largest singular value of the input.
Then, going one level higher,
the operator norm on $\mathcal L(\mathbb C^{n\times n})$ with respect to the trace norm
will be denoted $\|\cdot\|_{1\to 1}$, i.e.~$\|\Phi\|_{1\to 1}:=\sup_{A\in\mathbb C^{n\times n},\|A\|_1=1}\|\Phi(A)\|_1$,
while the symbol for the diamond norm (completely bounded trace norm) will be $\|\cdot\|_\diamond$, i.e.~$\|\Phi\|_\diamond:=\|\Phi\otimes{\rm id}_n\|_{1\to 1}$ \cite{Watrous18}.
Another norm of importance will be the ${\rm sup}$-norm: given a non-empty set $D$ and a map $f:D\to (X,\|\cdot\|_X)$ one defines $\|f\|_{\rm sup}:=\sup_{x\in D}\|f(x)\|_X$. For our purpose, the ${\rm sup}$-norm of maps $\Phi:D\subseteq\mathbb R\to\mathcal L(\mathbb C^{n\times n})$ will be evaluated with respect to
$\|\cdot\|_X=\|\cdot\|_\diamond$.
Finally, $\mathsf U(n)$ is the Lie group of all unitary $n\times n$ matrices, and $\mathfrak u(n)$ is its Lie algebra, i.e.~$i\mathfrak u(n)$ is the collection of all Hermitian $n\times n$ matrices.

As we saw in the introduction it will be vital to be precise about
what it means for a curve of unitaries $(U(t))_{t\geq 0}$---which are, e.g., not differentiable at zero---to describe system dynamics, i.e.~we want $U$ to satisfy
\begin{equation}\label{eq:schrodinger}
\dot U(t)=-iH(t)U(t)\qquad U(0)=\mathbbm1
\end{equation}
in some reasonable, yet general sense.
For this we go back to the fundamentals of differential equations and initial value problems.
The mathematical notion here---which is central to, e.g., control theory, as there one often works with 
piecewise constant 
functions
(i.e.~they cannot be the conventional derivative of some function)---is ``absolute continuity''.
Recall that a (locally) absolutely continuous function $f$ is a function which is differentiable almost everywhere and which can be reconstructed
from its derivative $\dot f$---which is (locally) integrable---via
$f(t)=f(t_0)+\int_{t_0}^t\dot f(\tau)\,{\rm d}\tau$.
The set of all locally integrable functions $f:I\subseteq\mathbb R\to V$ into a normed space $V$ will be denoted by $L^1_{\sf loc}(I,V)$.
Details can be found in Appendix~A.
Either way, combining this notion with Eq.~\eqref{eq:schrodinger} yields the integral version of the 
lifted Schr\"odinger equation:
\begin{equation}\label{eq:schrodinger-int}
U(t)=\mathbbm1-i\int_{0}^t H(\tau)U(\tau)\,{\rm d}\tau
\end{equation}
This leads to the arguably most general notion of dynamics of a closed quantum system:
\begin{defi}\label{def_closed_sys_dyn}
Given $t_f\in(0,\infty]$ a map $U:[0,t_f)\to\mathsf U(n)$ is said to describe 
\textit{closed system dynamics} if there exists $H:[0,t_f)\to i\mathfrak u(n)$ locally integrable
such that $U$ solves~\eqref{eq:schrodinger-int}
for all $t\in[0,t_f)$.
\end{defi}
\noindent Equivalently, closed system dynamics are precisely those curves of unitaries which are locally absolutely continuous and which take the value $\mathbbm1$ at zero.
For a more detailed treatment of 
such (so-called ``mild'') solutions 
of differential equations,
refer, again, to Appendix~A.\medskip

As we saw in the introduction, in this framework it is very much possible for dynamical semigroups to arise from closed system dynamics of a larger system when tracing out the environment.
Let us give such objects their own name, see also \cite{vE22_Stinespring}:

\begin{defi}\label{def_stinespring_curves}
Let $t_f\in(0,\infty]$ and $\Phi:[0,t_f)\to\mathsf{CPTP}(n)$ be given.
We say $\Phi$ is a \textit{Stinespring curve} if there exists $m\in\mathbb N$, a state $\omega\in\mathbb D(\mathbb C^{m})$, as well as
$U:[0,t_f)\to\mathsf U(mn)$ locally absolutely continuous
such that 
for all $t\in[0,t_f)$
\begin{equation*}
\Phi(t)= {\rm tr}_{\mathbb C^{m}}\big( U(t)((\cdot)\otimes\omega) U(t)^*\big)\,.
\end{equation*}
Here ${\rm tr}_{\mathbb C^{m}}:\mathbb C^{n\times n}\otimes\mathbb C^{m\times m}\to\mathbb C^{n\times n}$ is the usual partial trace.
Additionally, if $U$ is generated by a fixed (i.e.~time-independent) Hamiltonian $H\in i\mathfrak u(mn)$, i.e.~if $U(t)=e^{-iHt}U_0$ for some $U_0\in\mathsf U(mn)$ and
all $t\in[0,t_f)$, then we call $\Phi$ \textit{time-independent};
else we call the Stinespring curve \textit{time-dependent}.
\end{defi}

With this notation in place, the main result of Dive et al.~\cite{Dive15} reads as follows:
\begin{proposition}\label{prop_dive}
Let $n\in\mathbb N$, $t_f\in(0,\infty]$, and
$\Phi:[0,t_f)\to\mathsf{CPTP}(n)$ analytic be given.
Then $\Phi$ is a Stinespring curve with environment dimension at most $n^2$,
and the closed system dynamics $U$ of system plus environment can 
be chosen such that $U|_{(0,t_f)}$ is analytic. In particular, this holds for all $\Phi$ 
time-independent Markovian, i.e.~$\Phi(t)=  e^{tL}$ with $L$ of \textsc{gksl}-form \cite{GKS76,Lindblad76}.
\end{proposition}
\begin{remark}
As explained in the introduction, for Markovian dynamics the Hamiltonian of system plus environment has to diverge at zero.
Yet we want to mention that, in the past, systems with couplings which give rise to a bounded curve of Hamiltonians---thus necessarily resulting in non-Markovian dynamics---have been studied in the literature \cite{Chakraborty19}.
\end{remark}

While we are focusing on the environmental form of quantum channels in this work, let us reiterate that going to the level of Kraus operators yields an identical problem. This is due to the
equivalence of representations even in the dynamic setting, refer also to \cite[Ch.~3.2]{BreuPetr02}:

\begin{lemma}\label{lemma_equiv_repr}
Let $k,n\in\mathbb N$, $t_f\in(0,\infty]$ as well as $\Phi:[0,t_f)\to\mathsf{CPTP}(n)$ be given. The following statements are equivalent.
\begin{enumerate}[(i)] 
\item \label{lemma_equiv_repr_item_kraus} There exist functions $K_1,\ldots,K_k:[0,t_f)\to\mathbb C^{n\times n}$ locally absolutely continuous such that for all $t\in[0,t_f)$
$$
\Phi(t)= \sum_{j=1}^k K_j(t)(\cdot)K_j(t)^*\,.
$$
\item \label{lemma_equiv_repr_item_iso} There exists $V:[0,t_f)\to\mathbb C^{nk\times n}$ locally absolutely continuous such that for all $t\in[0,t_f)$, $V(t)$ is an isometry and
$$
\Phi(t)= {\rm tr}_{\mathbb C^k}\big(V(t)(\cdot)V(t)^*\big)\,.
$$
\item \label{lemma_equiv_repr_item_unitary} There exist an isometry $V\in\mathbb C^{nk\times n}$ as well as closed system dynamics $U:[0,t_f)\to\mathsf U(nk)$ such that for all $t\in[0,t_f)$
$$
\Phi(t)= {\rm tr}_{\mathbb C^k}\big(U(t)V(\cdot)V^*U(t)^*\big)\,.
$$
\end{enumerate}
Moreover, if $\Phi(0)={\rm id}$, then each of the above is equivalent to
\begin{itemize}
\item[(iii')] For all unit vectors $\psi\in\mathbb C^k$ (i.e.~$\|\psi\|=1$) there exist closed system dynamics $U:[0,t_f)\to\mathsf U(nk)$ such that for all $t\in[0,t_f)$
$$
\Phi(t)= {\rm tr}_{\mathbb C^k}\big(U(t)((\cdot)\otimes|\psi\rangle\langle\psi|)U(t)^*\big)\,.
$$
\end{itemize}
Finally, if any of the above objects is in some regularity class (e.g., continuously differentiable, smooth, analytic, etc.) on an open subinterval of $[0,t_f)$, then so are all the others.
\end{lemma}
\begin{proof}
``(\ref{lemma_equiv_repr_item_kraus}) $\Rightarrow$ (\ref{lemma_equiv_repr_item_iso})'': For all $t\in[0,t_f)$ define
$
V(t): \mathbb C^n \to \mathbb C^n\otimes\mathbb C^{k} $ via
$ x\mapsto\sum_{j=1}^{k}(K_j(t)x)\otimes |j\rangle\,.
$
Then (\ref{lemma_equiv_repr_item_kraus}) implies $\Phi(t)=  {\rm tr}_{\mathbb C^{k}}(V(t)(\cdot)V(t)^*)$ and  each $V(t)$ is an isometry (because $\sum_{j=1}^kK_j(t)^*K_j(t)=\mathbbm1$, due to $\Phi(t)$ being trace-preserving) for all $t\in[0,t_f)$ \cite[Thm.~6.9 \& Cor.~6.13]{Holevo12}.
As $V(t)$ is linear in each $K_j(t)$, local absolute continuity as well as any type of regularity transfers over from the $K_j$ to $V$.

``(\ref{lemma_equiv_repr_item_iso}) $\Rightarrow$ (\ref{lemma_equiv_repr_item_unitary})'': 
The task, essentially, is to complete $V(t)$ to a unitary $U(t)$ while maintaining the continuity requirement in $t$.
The idea for this is as follows:
Because $t\mapsto V(t)$ is 
locally absolutely continuous,
the ``generator'' $Q(t)=(\mathbbm1-V(t)V(t)^*)\dot V(t) V(t)^*-V(t)\dot V(t)^*(\mathbbm1-V(t)V(t)^*)$ is almost everywhere defined, locally integrable, and skew-Hermitian;
thus the solution to $\dot W(t)=Q(t)W(t)$, $W(0)=\mathbbm1$ exists and unitary.
Then one can show that in some (time-independent) basis the last $nk-n$ columns of $W(t)$ are precisely what completes $V(t)$ to a unitary $U(t)$ while preserving local absolute continuity in $t$.
The precise statement---proven in Appendix~B as Lemma~\ref{lemma_app_A}---reads as follows:
there exists
$U:[0,t_f)\to\mathsf U(nk)$ locally absolutely continuous
such that $V(t)=U(t)V(0)$, $U(0)=\mathbbm1$
for all $t\in[0,t_f)$.
Thus $U$ describes closed system dynamics, and (\ref{lemma_equiv_repr_item_iso})
implies (\ref{lemma_equiv_repr_item_unitary}) when defining $V:=V(0)$.
Moreover, the aforementioned lemma also guarantees that any regularity of $t\mapsto V(t)$ can be carried over to $U$.

``(\ref{lemma_equiv_repr_item_unitary}) $\Rightarrow$ (\ref{lemma_equiv_repr_item_kraus})'': Given any orthonormal basis $\{g_j\}_{j=1}^k$ of $\mathbb C^k$, it is well known that 
$\{\iota_j^*:j=1,\ldots,k\}$
are Kraus operators of 
${\rm tr}_{\mathbb C^k}$ 
 where $\iota_j:\mathbb C^n\to\mathbb C^n\otimes\mathbb C^k$ is defined via $\iota_j(x):= x\otimes g_j$, cf., e.g., \cite[Lemma B.2 \& C.1]{vE22_Stinespring}.
This implies that
$
\Phi(t)= {\rm tr}_{\mathbb C^k}\big(U(t)V(\cdot)V^*U(t)^*\big)=\sum_{j=1}^k\iota_j^*U(t)V(\cdot)V^*U(t)^*\iota_j
$
which shows (\ref{lemma_equiv_repr_item_kraus}) once we define $K_j(t):=\iota_j^*U(t)V\in\mathbb C^{n\times n}$ for all $j=1,\ldots,k$, $t\in[0,t_f)$.
Again, each $K_j$ is linear in $U(t)$ so local absolute continuity
as well as any regularity carries over from $U$ to the $K_j$.

Finally, assume $\Phi(0)={\rm id}$. While ``(\ref{lemma_equiv_repr_item_unitary}) $\Leftarrow$ (iii')'' follows directly from 
the relation $(\cdot)\otimes|\psi\rangle\langle\psi|=\iota_\psi(\cdot)\iota_\psi^*$ (where $\iota_\psi(x):=x\otimes\psi$),
cf.~\cite[Lemma~C.1]{vE22_Stinespring}, for ``(\ref{lemma_equiv_repr_item_unitary}) $\Rightarrow$ (iii')'' 
note that
${\rm tr}_{\mathbb C^k}( V (\cdot) V^* )=\Phi(0)={\rm id}={\rm tr}_{\mathbb C^k}( \iota_\psi (\cdot) \iota_\psi^* )$.
This means that $V$, $\iota_\psi$ are locally unitarily equivalent:  Coro.~2.24 in \cite{Watrous18} yields $W\in\mathsf U(k)$ such that $V=(\mathbbm1\otimes W)\iota_\psi$, i.e.~$Vx=x\otimes W\psi=x\otimes\phi$
for all $x\in\mathbb C^n$
 where $\phi:=W\psi\in\mathbb C^k$. 
Therefore $VA V^*=A\otimes|\phi\rangle\langle\phi|$ for all $A\in\mathbb C^{n\times n}$, that is,
(iii') holds for $\psi=\phi$.
Now if $\psi\in\mathbb C^k$ is an arbitrary unit vector there certainly exists $W'\in\mathsf U(k)$ such that $\phi=W'\psi$;
all we then need is the modification $U(t)\to (\mathbbm1\otimes W')^*U(t)(\mathbbm1\otimes W')$.
Obviously, the latter still describes closed system dynamics and we compute
\begin{align*}
\Phi(t)&={\rm tr}_{\mathbb C^{k}}\big(U(t)((\cdot)\otimes|\phi\rangle\langle\phi|)U(t)^*\big)\\
&={\rm tr}_{\mathbb C^{k}}\big(U(t)((\cdot)\otimes W'|\psi\rangle\langle\psi|W'^*)U(t)^*\big)\\
&={\rm tr}_{\mathbb C^{k}}\big(U(t)(\mathbbm1\otimes W')((\cdot)\otimes|\psi\rangle\langle\psi|)(\mathbbm1\otimes W')^*U(t)^*\big)\\
&={\rm tr}_{\mathbb C^{k}}\big((\mathbbm1\otimes W')^*U(t)(\mathbbm1\otimes W')((\cdot)\otimes|\psi\rangle\langle\psi|)(\mathbbm1\otimes W')^*U(t)^*(\mathbbm1\otimes W')\big)
\end{align*}
for all $t\in[0,t_f)$.
\end{proof}

\begin{remark}\label{rem_cont_ac}
Lemma~\ref{lemma_equiv_repr} continues to hold
if local absolute continuity is replaced by usual continuity. 
The only thing that changes
about the above proof 
is the step where one extends $V(t)$ to a continuous curve of unitaries $U(t)$:
the problem is that if $V(t)$ is only continuous, then we cannot guarantee that the
generator $Q(t)=(\mathbbm1-V(t)V(t)^*)\dot V(t) V(t)^*-V(t)\dot V(t)^*(\mathbbm1-V(t)V(t)^*)$---which was the key to constructing $U(t)$---even exists.
One can circumvent this problem via an old result of Dole\v{z}al \cite{Dolezal64} which essentially guarantees existence of such an extension $U(t)$ in the continuous case (cf.~also Rem.~\ref{rem_cont_ac_app} in Appendix~B), the drawback being that the construction is not as explicit as the one involving $Q(t)$.
However, in light of Def.~\ref{def_closed_sys_dyn} it could be argued that
mere continuity is too weak for describing and modeling physical systems either way.
\end{remark}

\section{Main Result}\label{sec:main}

With Definitions~\ref{def_closed_sys_dyn} \& \ref{def_stinespring_curves} in place we are finally ready to state our main result:

\begin{thm}\label{thm0}
Let $n\in\mathbb N$, $t_f\in(0,\infty]$, as well as $\Phi:[0,t_f)
\to\mathsf{CPTP}(n)$ Lipschitz
with constant $K_\Phi>0$---with respect to, say, the diamond norm\footnote{
so $\Phi$ satisfies $\|\Phi(t_1)-\Phi(t_2)\|_\diamond\leq K_\Phi|t_1-t_2|$ for all $t_1,t_2\in[0,t_f)$
}---be given.
Then for all $\varepsilon>0$ there exists a Stinespring curve $\Phi_\varepsilon$ on $[0,t_f)$ with dilation space $\mathbb C^{n^2}\otimes\mathbb C^2\otimes\mathbb C^2$ such that
$
\|\Phi-\Phi_\varepsilon\|_{\rm sup}<\varepsilon
$.
Moreover, 
\begin{enumerate}[(i)]
\item \label{thm0_item_step} the time-dependent Hamiltonian
$H\in L^1_{\sf loc}([0,t_f),i\mathfrak u(4n^3))$ which
generates the dynamics of the larger 
closed system of $\Phi_\varepsilon$ can be
chosen piecewise constant with $\|H\|_{\rm sup}=\sup_{t\in[0,t_f)}\|H(t)\|_\infty<\infty$.
\item \label{thm0_item_analytic} if $t_f<\infty$, then the Hamiltonian $H:[0,t_f)\to i\mathfrak u(4n^3)$ can even be chosen analytic.
\item \label{thm0_item_aux_pure} the auxiliary state of $\Phi_\varepsilon$ can be chosen pure.
\item \label{thm0_item_csd} if $\Phi(0)={\rm id}$, then the curve of unitaries $t\mapsto U(t)$ which generates $\Phi_\varepsilon$ can be chosen such that $U(0)=\mathbbm1$, i.e.~$U$ describes closed system dynamics.
\end{enumerate}
\end{thm}
\begin{proof}
Given any $\varepsilon>0$, the idea for constructing $\Phi_\varepsilon$ is to 
consider a sufficiently discretized version of $\Phi$ and to ``connect the dots'' appropriately.
More precisely, choose $\delta\in(0,\min\{1,t_f\})$ such that 
$\delta<\varepsilon^2(K_\Phi+4\pi\sqrt{K_\Phi})^{-2}$, and
consider
$\{\Phi(j\delta):j\in\mathbb N_0,j<M_\delta\}$ where $M_\delta:=\lceil \frac{t_f}{\delta}\rceil$.
Next, for each $j<M_\delta$ pick an arbitrary Stinespring isometry $V_j:\mathbb C^n\to\mathbb C^n\otimes\mathbb C^{2n^2}$ of $\Phi(j\delta)$, i.e.~$\Phi(j\delta)={\rm tr}_{\mathbb C^{2n^2}}(V_j(\cdot)V_j^*)$.
The key to our construction is 
the following continuity result for Stinespring dilations:
for any $\Psi_1,\Psi_2\in\mathsf{CPTP}(n)$ and any respective Stinespring isometries $Z_1,Z_2:\mathbb C^n\to\mathbb C^n\otimes\mathbb C^{2n^2}$ of $\Psi_1,\Psi_2$, there exists $W\in\mathsf U(2n^2)$ such that\footnote{
Recall that $(\mathbbm1\otimes W)Z_2$ is again a Stinespring isometry, regardless of $W\in\mathsf U(2n^2)$.
}
$\|Z_1-(\mathbbm1\otimes W)Z_2\|_\infty\leq\sqrt{\| \Psi_1-\Psi_2 \|_\diamond}$, cf.~\cite[Thm.~1]{Kretschmann08}, \cite[Prop.~5]{DAriano07}, or \cite[Prop.~1]{vomEnde_KSW_23}.
Starting from $W_0:=\mathbbm1$,
an inductive application of this result to $\Phi(j\delta)$, $\Phi((j+1)\delta)$
yields $W_{j+1}\in\mathsf U(2n^2)$ such that
\begin{equation}\label{eq:dist_U_upperbound}
\|(\mathbbm1\otimes W_j)V_j-(\mathbbm1\otimes W_{j+1})V_{j+1}\|_\infty\leq \sqrt{\|\Phi(j\delta)-\Phi((j+1)\delta)\|_\diamond}\,.
\end{equation}
for all $j\in\mathbb N_0$, $j<M_\delta$.
In other words this gives rise to a new family of Stinespring isometries $\{V_0,(\mathbbm1\otimes W_1)V_1,(\mathbbm1\otimes W_2)V_2,\ldots\}$ of $\{\Phi(0),\Phi(\delta),\Phi(2\delta),\ldots\}$, respectively, where consecutive elements of the former set are as close together as the (square root of the) corresponding channel distance.
Next, in order to connect all $\Phi(j\delta)$ via a Stinespring curve we have to extend each Stinespring isometry $(\mathbbm1\otimes W_j)V_j$ to a Stinespring unitary.
For this we define\footnote{
More precisely in~\eqref{eq:def_Uj_SzNagy}, $U_j\in\mathcal L(\mathbb C^n\otimes\mathbb C^{2n^2}\otimes\mathbb C^2)$
is defined via
$$
U_j:=J\circ\begin{pmatrix}
(\mathbbm1\otimes W_j)V_j\iota^*&(\mathbbm1\otimes W_j)(\mathbbm1-V_jV_j^*)(\mathbbm1\otimes W_j)^*\\
\mathbbm1-\iota\iota^*&-\iota V_j^*(\mathbbm1\otimes W_j)^*
\end{pmatrix}\circ J^{-1}
$$
where 
$\circ$ is the usual composition of maps, and
$J:(\mathbb C^n\otimes\mathbb C^{2n^2})\times(\mathbb C^n\otimes\mathbb C^{2n^2})\to\mathbb C^n\otimes\mathbb C^{2n^2}\otimes\mathbb C^2$ is the isometric isomorphism
$J(x,y):=x\otimes|0\rangle+y\otimes |1\rangle$.
However, in abuse of notation we will use block matrices (e.g., in~\eqref{eq:def_Uj_SzNagy}) and the corresponding (induced) operator on the tensor product interchangeably.
}
\begin{equation}\label{eq:def_Uj_SzNagy}
U_j:=\begin{pmatrix}
(\mathbbm1\otimes W_j)V_j\iota^*&(\mathbbm1\otimes W_j)(\mathbbm1-V_jV_j^*)(\mathbbm1\otimes W_j)^*\\
\mathbbm1-\iota\iota^*&-\iota V_j^*(\mathbbm1\otimes W_j)^*
\end{pmatrix}\in\mathbb C^{4n^3\times 4n^3}\,,
\end{equation}
where $\iota:\mathbb C^n\to\mathbb C^n\otimes\mathbb C^{2n^2}$ is the injective embedding $\iota(x):=x\otimes|0\rangle$.
In particular, $\iota$ is an isometry (i.e.~$\iota^*\iota=\mathbbm1$), which---after a straightforward computation---implies that $U_j$ is unitary.
Moreover, this 
$U_j$ for all $x\in\mathbb C^n$ satisfies $U_j(x\otimes|0\rangle\otimes|0\rangle)=((\mathbbm1\otimes W_j)V_jx)\otimes|0\rangle$
because of the following identity: $(\mathbbm1-\iota\iota^*)(x\otimes|0\rangle)=(\mathbbm1-\mathbbm1\otimes|0\rangle\langle 0|)(x\otimes|0\rangle)=0$.
Altogether, this shows that $U_j$ is
a Stinespring unitary of $\Phi(j\delta)$ w.r.t.~the environment state
$|0\rangle\langle 0|_{4n^2}=|0\rangle\langle 0|\otimes|0\rangle\langle 0|\in\mathbb C^{2n^2\times 2n^2}\otimes\mathbb C^{2\times 2}$:
\begin{align*}
{\rm tr}_{\mathbb C^{4n^2}}(&U_j((\cdot)\otimes|0\rangle\langle 0|\otimes|0\rangle\langle 0|)U_j^*)\\
&=
{\rm tr}_{\mathbb C^{4n^2}}((\mathbbm1\otimes W_j)V_j(\cdot)V_j^*(\mathbbm1\otimes W_j)^*\otimes|0\rangle\langle 0|)\\
&=
{\rm tr}_{\mathbb C^{2n^2}}((\mathbbm1\otimes W_j)V_j(\cdot)V_j^*(\mathbbm1\otimes W_j)^*)={\rm tr}_{\mathbb C^{2n^2}}(V_j(\cdot)V_j^*)=\Phi(j\delta)
\end{align*}
Now that we have a discrete set $\{U_j:j\in\mathbb N_0,j<M_\delta\}$
of Stinespring unitaries
we have to turn it into a curve of unitaries $U(t)$; however, not any such curve will do as we have to be careful that the norm of the generating Hamiltonian 
does not become ``too large''.
More precisely, for all $j=1,\ldots,M_\delta-1$ there exists $H_j\in i\mathfrak u(4n^3)$ such that $U_{j}U_{j-1}^*=e^{iH_j}$ and $2\|H_j\|_\infty\leq \pi\|U_{j}U_{j-1}^*-\mathbbm1\|_\infty$; the existence of such $H_j$ is proven in Lemma~\ref{lemma_app_C_unitary_gen}~(ii) in Appendix~C.
With this, define a piecewise constant path
$
{\sf H}:[0,t_f)\to i\mathfrak u(4n^3)
$
via 
$
{\sf H}(t):=-\delta^{-1} H_j
$
for all $t\in[(j-1)\delta,j\delta)$, $j<M_\delta$ and $0$ else.
Note that, obviously, $\mathsf H$ is locally integrable and Hermitian at all times so
$\dot U(t)=-i{\sf H}(t)U(t)$, $U(0)=U_0$
has a (unique, locally absolutely continuous) solution
$U:[0,t_f)\to\mathsf U(4n^3)$.
As desired, $U(t)$ ``connects'' the discrete Stinespring unitaries $U_j$:
\begin{align*}
U(j\delta)&=\Big(\prod_{\alpha=1}^j e^{-i[\alpha\delta-(\alpha-1)\delta]{\sf H}((\alpha-1)\delta)}\Big)U(0)\\
&=\Big(\prod_{\alpha=1}^j e^{iH_\alpha}\Big)U_0
=\Big(\prod_{\alpha=1}^j U_\alpha U_{\alpha-1}^*\Big)U_0=U_j
\end{align*}
for all $j<M_\delta$.
With this we define $\Phi_\varepsilon:[0,t_f)\to\mathsf{CPTP}(n)$ via
$$
\Phi_\varepsilon(t):= {\rm tr}_{\mathbb C^{4n^2}}\big(U(t)((\cdot)\otimes|0\rangle\langle 0|)U(t)^*\big)
$$
for all $t\in[0,t_f)$, i.e.~$\Phi_\varepsilon$ is a time-dependent Stinespring curve and (\ref{thm0_item_aux_pure}) holds.
We claim that this $\Phi_\varepsilon$ is the Stinespring curve we were looking for, i.e.~that $
\|\Phi-\Phi_\varepsilon\|_{\rm sup}<\varepsilon$.
For this the idea is as follows: $\Phi$, $\Phi_\varepsilon$ coincide at times $t=0,\delta,2\delta,\ldots$ so $\|\Phi-\Phi_\varepsilon\|_{\rm sup}$ 
can be upper bounded by the product of the step-size $\Delta t=\delta$ and the sum of any Lipschitz constants of $\Phi$, $\Phi_\varepsilon$.
Intuitively, the Lipschitz constants determine how much $\Phi_\varepsilon$ can deviate from $\Phi$ between $t=(j-1)\delta$ and $t=j\delta$.
Thus, the first step is to check that $\Phi_\varepsilon$ is 
Lipschitz.
Let $j<M_\delta$, $t\in((j-1)\delta,j\delta)$ be given.
Then $U$ is differentiable with $\dot U(t)=i\delta^{-1}H_jU(t)$ which implies that $\Phi_\varepsilon$ is differentiable at $t$ with
\begin{align*}
\dot\Phi_\varepsilon(t)=i\delta^{-1}{\rm tr}_{\mathbb C^{4n^2}}\big(\big[H_j,U(t)\big((\cdot)\otimes|0\rangle\langle 0|\big)U(t)^*\big]\big)\,,
\end{align*}
i.e.~$\dot\Phi_\varepsilon(t)=i\delta^{-1}{\rm tr}_{\mathbb C^{4n^2}}\circ [H_j,\cdot]\circ {\rm Ad}_{U(t)}\circ \iota_{|0\rangle\langle 0|}$ where, again, $\circ$ is the usual composition of maps, ${\rm Ad}_U:=U(\cdot)U^*$ for any $U\in\mathsf U(n)$, and $\iota_{\omega}:=(\cdot)\otimes\omega$
for any $\omega\in\mathbb C^{n\times n}$.
Now the identity $\|\Psi_1\circ\Psi_2\|_\diamond\leq\|\Psi_1\|_\diamond\|\Psi_2\|_\diamond$ for all $\Psi_1,\Psi_2\in\mathcal L(\mathbb C^{n\times n})$ \cite[Prop.~3.48]{Watrous18}
together with the well-known fact that every $\mathsf{CPTP}$ map is $\|\cdot\|_{\diamond}$-contractive
\cite[Prop.~3.44]{Watrous18}
lets us compute
\begin{align*}
\|\dot\Phi_\varepsilon(t)\|_{\diamond}\leq\delta^{-1}\|[H_j,\,\cdot\,]\|_{\diamond}&=\delta^{-1}\|[H_j,\,\cdot\,]\otimes{\rm id}_n\|_{1\to 1}\\
&=\delta^{-1}\|[H_j\otimes\mathbbm1,\,\cdot\,]\|_{1\to 1}\\
&\leq 2\delta^{-1}\|H_j\otimes\mathbbm 1\|_\infty=2\delta^{-1}\|H_j\|_\infty\,.
\end{align*}
Because $\Phi_\varepsilon$ is continuous and differentiable almost everywhere, 
the largest value of the norm of the derivative is a Lipschitz constant for $\Phi_\varepsilon$ (cf.~Lemma~\ref{lemma_A2} in Appendix~C) 
meaning 
$\Phi_\varepsilon$ is Lipschitz-continuous with Lipschitz constant
\begin{align*}
\sup_{j\in\mathbb N,j<M_\delta} 2\delta^{-1}\|H_j\|_\infty
&\overset{\hphantom{\eqref{eq:dist_U_upperbound}}}\leq\sup_{j\in\mathbb N,j<M_\delta} \pi\delta^{-1}\|U_j-U_{j-1}\|_\infty\\
&\overset{\hphantom{\eqref{eq:dist_U_upperbound}}}\leq4\pi\delta^{-1}\sup_{j\in\mathbb N,j<M_\delta} \|(\mathbbm1\otimes W_j)V_j-(\mathbbm1\otimes W_{j+1})V_{j+1}\|_\infty\\
&\overset{\eqref{eq:dist_U_upperbound}}\leq4\pi\delta^{-1}\sup_{j\in\mathbb N,j<M_\delta}  \sqrt{\|\Phi(j\delta)-\Phi((j-1)\delta)\|_\diamond}\leq4\pi\sqrt{\tfrac{K_\Phi}\delta}
\end{align*}
Here, in the second step we used 
$
\|U_j-U_k\|_\infty\leq4\|(\mathbbm1\otimes W_j)V_j-(\mathbbm1\otimes W_k)V_k\|_\infty
$
for all $j,k<M_\delta$---which follows from the definition~\eqref{eq:def_Uj_SzNagy} of $U_j$ via a straightforward computation---and in the last step we used Lipschitz-continuity of $\Phi$.
Either way we now know that 
both $\Phi,\Phi_\varepsilon$ are Lipschitz-continuous on $[0,t_f)$ and they---by
construction---for all $j<M_\delta$ satisfy
\begin{align*}
\Phi_\varepsilon(j\delta)&={\rm tr}_{\mathbb C^{4n^2}}(U(j\delta)((\cdot)\otimes|0\rangle\langle 0|)U(j\delta)^*)\\
&= {\rm tr}_{\mathbb C^{4n^2}}\big(U_j((\cdot)\otimes|0\rangle\langle 0|)U_j\big)=\Phi(j\delta)\,.
\end{align*}
As stated previously, this lets us upper bound $\|\Phi-\Phi_\varepsilon\|_{\rm sup}$ by $\delta$
times the sum of the Lipschitz constants (cf.~Lemma~\ref{lemma_A1} in Appendix~B for precise statement and proof).
Recalling that $\delta<1$ was chosen less than $\varepsilon^2(K_\Phi+4\pi\sqrt{K_\Phi})^{-2}$,
\begin{align*}
\|\Phi-\Phi_\varepsilon\|_{\rm sup}&\leq \delta(K_\Phi+4\pi\delta^{-1/2}\sqrt{K_\Phi})
\leq \sqrt{\delta}(K_\Phi+4\pi\sqrt{K_\Phi})<\varepsilon\,.
\end{align*}

Now let us quickly reflect on which statements we have yet to prove.
We constructed a Stinespring curve $\Phi_\varepsilon$ with pure auxiliary state $|0\rangle\langle 0|$ (statement~\eqref{thm0_item_aux_pure}) which is $\varepsilon$-close to $\Phi$, and the time-dependent Hamiltonian $\mathsf H$
that generates the larger unitary dynamics of $\Phi_\varepsilon$ is piecewise constant with $\|\mathsf H\|_{\rm sup}=\sup_{j\in\mathbb N,j<M_\delta} \delta^{-1}\|H_j\|_\infty\leq 2\pi\sqrt{K_\Phi\delta^{-1}}<\infty$ (statement~\eqref{thm0_item_step}).
Hence (\ref{thm0_item_analytic}) and (\ref{thm0_item_csd}) are still open.

The key to proving (\ref{thm0_item_analytic})---i.e.~if $t_f<\infty$, then $\mathsf H$ can be chosen analytic---is, unsurprisingly, the fact that polynomials are dense in $(L^1([0,t_f]),\|\cdot\|_1)$ (density of the continuous functions in $L^1$ \cite[Thm.~3.14]{Rudin86} together with the fact that polynomials are dense in the continuous functions by Stone-Weierstrass).
Starting from $t_f>0$, a Lipschitz function $\Phi$, as well as $\varepsilon>0$, (\ref{thm0_item_step}) yields $\mathsf H:[0,t_f]\to i\mathfrak u(4n^3)$ piecewise constant and uniformly bounded such that the induced Stinespring curve $\Phi_\varepsilon$ is $\frac{\varepsilon}{2}$-close to $\Phi$ in sup-norm.
To construct an analytic Hamiltonian $\tilde{\mathsf H}$ which approximates $\mathsf H$ we will simply approximate the matrix elements of $\mathsf H$.
Indeed, define $H_{jj}:[0,t_f]\to \mathbb R$ for all $j=1,\ldots,4n^3$ via $H_{jj}(t):=\langle j|\mathsf H(t)|j\rangle$, and define $H_{jk}:[0,t_f]\to\mathbb C$ 
for all $1\leq j<k\leq 4n^3$ via $H_{jk}(t):=\langle j|\mathsf H(t)|k\rangle$. By the above density argument there exist polynomials $\{\tilde H_{jj}:[0,t_f]\to\mathbb R\,:\,j=1,\ldots,4n^3\}$, $\{\tilde H_{jk}:[0,t_f]\to\mathbb C\,:\,1\leq j<k\leq 4n^3\}$ such that $\|H_{jk}-\tilde H_{jk}\|_{1}<\tfrac{\varepsilon}{64n^6}$ for all $j\leq k$. This lets us define the analytic Hamiltonian
\begin{align*}
\tilde{\mathsf{H}}:[0,t_f]&\to i\mathfrak u(4n^3)\\
t&\mapsto \sum_{j=1}^{4n^3}\tilde H_{jj}(t)|j\rangle\langle j|+\sum_{j,k=1,j< k}^{4n^3}\big(\tilde H_{jk}(t)|j\rangle\langle k|+\tilde H_{jk}^*(t)|k\rangle\langle j|\big)\,.
\end{align*}
A crude estimate shows $\|\mathsf H-\tilde{\mathsf{H}}\|_{1}
<\tfrac{(4n^3)^2\varepsilon}{64n^6}=\tfrac{\varepsilon}4$.
This carries over to the solutions of the corresponding differential equations $\dot U(t)=-i{\sf H}(t)U(t)$, $U(0)=U_0$ and $\dot {\tilde U}(t)=-i{\tilde{\sf H}}(t)\tilde U(t)$, $\tilde U(0)=U_0$ via the identity
\begin{equation}\label{eq:duhamel_U}
U(t)-\tilde U(t)=i\tilde U(t)\int_0^t\tilde U(\tau)^*(\tilde{\mathsf H}(\tau)-\mathsf H(\tau)) U(\tau)\,{\rm d}\tau
\end{equation}
which holds for all $t\in[0,t_f)$ \cite[Ch.~1, Thm.~5.1 \& Ch.~1.8]{DF84}.
Together with the fact that unitary matrices have operator norm $1$, Eq.~\eqref{eq:duhamel_U} implies
$\|U-\tilde U\|_{\rm sup}\leq \|\mathsf H-\tilde{\mathsf H}\|_1$.
Altogether, 
$\tilde\Phi_\varepsilon(t):= {\rm tr}_{\mathbb C^{4n^2}}(\tilde U(t)((\cdot)\otimes|0\rangle\langle 0|)\tilde U(t)^*) $, $t\in [0,t_f]$ 
is an analytic Stinespring curve which satisfies
\begin{align*}
 \|\Phi-\tilde\Phi_\varepsilon\|_{\rm sup}&\leq  \|\Phi-\Phi_\varepsilon\|_{\rm sup}+ \|\Phi_\varepsilon-\tilde\Phi_\varepsilon\|_{\rm sup}\\
&<\tfrac{\varepsilon}{2}+\big\| {\rm tr}_{\mathbb C^{4n^2}}((U-\tilde U)((\cdot)\otimes|0\rangle\langle 0|) U^*) \big\|_{\rm sup}\\
&\hphantom{<\tfrac{\varepsilon}{2}}\;+ \big\| {\rm tr}_{\mathbb C^{4n^2}}(\tilde U((\cdot)\otimes|0\rangle\langle 0|)(U-\tilde U)^*)  \big\|_{\rm sup}\\
&\leq\tfrac{\varepsilon}2+2\|U-\tilde U\|_{\rm sup}\leq \tfrac{\varepsilon}2+2 \|\mathsf H-\tilde{\mathsf H}\|_1<  \varepsilon
  \,;
\end{align*}
this shows (\ref{thm0_item_analytic}).

Finally let us prove (\ref{thm0_item_csd}), i.e.~if $\Phi(0)={\rm id}$, then $U$ can be chosen such that $U(0)=\mathbbm1$.
As in the proof of Lemma~\ref{lemma_equiv_repr} (\ref{lemma_equiv_repr_item_unitary}) $\Rightarrow$ (iii'),
because the Stinespring isometry $V_0$ satisfies ${\rm tr}_{\mathbb C^{2n^2}}(V_0(\cdot)V_0^*)=\Phi(0)={\rm id}={\rm tr}_{\mathbb C^{2n^2}}((\cdot)\otimes|0\rangle\langle 0|)$
there exists $Z\in\mathsf U(2n^2)$ such that
$V_0x=x\otimes Z|0\rangle$ \cite[Coro.~2.24]{Watrous18};
in other words $V_0x=x\otimes\psi$ 
for all $x\in\mathbb C^n$ where $\psi:=Z|0\rangle\in\mathbb C^{2n^2}$.
Then $U_0$ from~\eqref{eq:def_Uj_SzNagy} simplifies considerably:
because $W_0=\mathbbm1$, $V_0=(\mathbbm1\otimes Z)\iota$, and $\iota\iota^*=\mathbbm1\otimes|0\rangle\langle 0|$ a straightforward computation shows that $U_0=\mathbbm1\otimes U_A$ where
$$
U_A=\begin{pmatrix}
|\psi\rangle\langle 0|&\mathbbm1-|\psi\rangle\langle\psi|\\
\mathbbm1-|0\rangle\langle 0|&-|0\rangle\langle\psi|
\end{pmatrix}\in\mathsf U(4n^2)\,.
$$
With this we define a new curve
$\tilde U:[0,t_f)\to\mathsf U(4n^3)$
via 
$\tilde U(t):=U(t)(\mathbbm 1\otimes U_A^*)$
with $U(t)$ from above.
Also we define
$\Phi_\varepsilon:[0,t_f)\to\mathsf{CPTP}(n)$ via
\begin{align*}
\Phi_\varepsilon(t):= \;&{\rm tr}_{\mathbb C^{4n^2}}(\tilde U(t)((\cdot)\otimes|\psi\rangle\langle \psi|\otimes|0\rangle\langle 0|)\tilde U(t)^*)\\
=\;&{\rm tr}_{\mathbb C^{4n^2}}\big( U(t)\big((\cdot)\otimes U_A^*(|\psi\rangle\langle \psi|\otimes|0\rangle\langle 0|) U_A\big) U(t)^*\big)\\
=\;&{\rm tr}_{\mathbb C^{4n^2}}(U(t)((\cdot)\otimes|0\rangle\langle 0|\otimes|0\rangle\langle 0|)U(t)^*)
\end{align*}
for all $t\in[0,t_f)$.
The last rearrangement implies that $\Phi_\varepsilon$ is $\varepsilon$-close to $\Phi$ in ${\rm sup}$-norm.
Moreover, $\tilde U$ is locally absolutely continuous with $\tilde U(0)=\mathbbm1$ (i.e.~$\tilde U$ describes closed system dynamics).
This concludes the proof.
\end{proof}
\noindent We emphasize that, in the case of (\ref{thm0_item_csd}) (i.e.~if $\Phi(0)={\rm id}$),
statements (\ref{thm0_item_step}) through (\ref{thm0_item_aux_pure}) continue to hold, and that if the domain of $\Phi$ is a compact time interval, then Lipschitz continuity can be relaxed to local Lipschitz continuity, cf.~\cite[Ch.~1.4]{Burago01} \& \cite[p.~93]{Kumaresan05}.\medskip

We conclude this section with some remarks on the dimension of the auxiliary system in Thm.~\ref{thm0}:
\begin{remark}
\begin{itemize}
\item[(i)] A direct consequence of our proof is that---just like in the static case---the $n^2$ in the environment dimension can be lowered to the largest Kraus rank of $\Phi(t)$ taken over all $t$.
\item[(ii)] One additional factor of two comes from our choice of dilation~\eqref{eq:def_Uj_SzNagy}.
This extra environment qubit can probably be avoided by using the construction from Appendix~B where a dynamic isometry is turned to something unitary by adding suitable columns (and not via a larger block matrix).
However, taking this route makes it is more difficult to keep track of the error from~\eqref{eq:dist_U_upperbound}.
\item[(iii)] The other additional factor of two comes from the continuity result for Stinespring isometries we used.
Getting rid of this extra qubit would amount to improving said result to hold for all $m\geq n^2$ (instead of, currently, just all $m\geq 2n^2$), probably at the cost of enlarging the right-hand side of~\eqref{eq:dist_U_upperbound} by a factor or $\sqrt2$, refer also to \cite{vomEnde_KSW_23}.
\end{itemize}
\end{remark}

Moreover, combining our main theorem with Lemma~\ref{lemma_equiv_repr}
shows that for all quantum dynamics there exist at most $4n^2$ dynamic (locally absolutely continuous, and analytic for bounded time intervals) Kraus operators which uniformly approximate the dynamics in question.
\section{Outlook}
In this paper, we tackled the question of whether the dynamics of any quantum system admit a continuous (resp.~sufficiently regular) curve of Stinespring unitaries.
Or, reformulating this physically: can the evolution of an open system always be lifted to dynamics of a larger closed system while adding only finitely many degrees of freedom?
Answering this question---ideally in a constructive way---could 
simplify simulating, and, in general, studying open systems, cf.~also \cite{Dive15}.

Our contribution to this problem is the result that, given any Lipschitz continuous dynamics such time-dependent Stinespring unitaries exist if one allows for arbitrarily small errors in the supremum norm.
Moreover, for finite times these unitaries can even be chosen analytic.
One way to interpret our main result is that there is no (substantial) ``gap'' between
general quantum dynamics and time-dependent Stinespring curves
(resp.~dynamic Kraus operators).
Taking this perspective comes with a number of follow-up questions---even beyond the obvious one which is whether
uniform approximation can be replaced by exactness.
\begin{itemize}
\item What is the role of time-dependence of the system-environment Hamiltonian?
More precisely, given\footnote{
We stress that boundedness of the interval is a necessary assumption.
This is a direct
consequence of the quantum recurrence theorem \cite{Bocchieri57,Schulman78} (cf.~also \cite{wallace2015recurrence,keyl18InfLie}), i.e.~every finite-dimensional
time-\textit{independent} Stinespring curve eventually revisits the identity (at least approximately), but a generic dynamic process does not do that.
}
$\Phi:[0,t_f]\to\mathsf{CPTP}(n)$, $t_f>0$ locally Lipschitz with $\Phi(0)={\rm id}$, can
$\Phi$ be approximated by a (finite-dimensional) time-\textit{independent} Stinespring curve to arbitrary degree?
\end{itemize}
\noindent At first glance, this may appear simple: first apply Thm.~\ref{thm0} to get an approximation $\Phi_\varepsilon$ generated by (analytic) closed system dynamics $U(t)$.
Then---motivated by \cite{Burgarth22} where an infinite-dimensional clock was used to construct an autonomous unitary dilation of a dynamical semigroup consistent with dynamical decoupling---given a sufficiently small step size $\delta$, collect $U(0)$, $U(\delta)$, $U(2\delta),$ etc., into a larger unitary via a finite clock system.
More precisely, define $
W:=\sum_{j=1}^\ell U(j\delta)U((j-1)\delta )^*\otimes|e_{j+1}\rangle\langle e_j|
$
to recover the curve $\Phi_\varepsilon$ at each time-step $j\delta$ after tracing out:
$$
\Phi_\varepsilon(j\delta)=
{\rm tr}_E\big(W^j((\cdot)\otimes|0\rangle\langle 0|\otimes|0\rangle\langle 0|)(W^*)^j\big)\,.
$$
When writing $W$ as $e^{-iH}$ for suitable $H$ this yields
a time-independent Stinespring curve 
$\Phi_\varepsilon'(t):={\rm tr}_E(e^{-iHt}((\cdot)\otimes|0\rangle\langle 0|\otimes|0\rangle\langle 0|)e^{iHt})$;
however,
it is not obvious whether the error $\|\Phi_\varepsilon(t)-\Phi_\varepsilon'(t)\|_{\diamond}$ ``produced'' between $j\delta$ and $(j+1)\delta$ (more precisely: the Lipschitz constant of $\Phi_\varepsilon'$)
is of order $o(\delta^{-1})$.
This would be necessary for the proof of Thm.~\ref{thm0} to generalize to this setting as
``the'' Lipschitz constant of
$\Phi_\varepsilon'$ is decided by $\|H\|_\infty$.
However, this is anything but straightforward as the added clock makes $W$ non-block-diagonal.
\begin{itemize}
\item Can the exactness result of Dive et al.~\cite{Dive15} be generalized to time-\textit{dependent} Markovian dynamics? While their result covers the analytic case as well as the case
where one takes \textit{finite} products of time-independent Markovian maps, taking the closure over the latter \cite{Wolf08a} is what makes the step from time-independent to time-dependent non-trivial.
\end{itemize}
This question is also inspired by the idea of a ``Markovian
shell''
\cite{Garraway97,KR09,CK10,JPB_decoh,CK16,TSHP18,COK23}
where one wants
to express arbitrary, but fixed dynamics $\Phi$
via
$
\Phi(t)= \big({\rm tr}_{\mathbb C^{m}}\circ\Psi_t\big)((\cdot)\otimes\omega)
$,
for some $m\in\mathbb N$, $\omega\in\mathbb D(\mathbb C^{m})$, as well as $\Psi:[0,t_f)\to\mathsf{CPTP}(mn)$ time-dependent Markovian\footnote{
Recall that a map $\Psi:I\to\mathsf{CPTP}(n)$ is called time-dependent Markovian if it is the solution to a time-local master equation $\dot\Psi(t)=L(t)\Psi(t)$, $\Psi(0)={\rm id}$ (resp.~the corresponding integral equation) for some $L:[0,t_f)\to\mathcal L(\mathbb C^{n\times n})$ locally integrable with values in the \textsc{gksl}-generators, refer to \cite{Wolf08a,DHKS08,OSID_thermal_res}.
}.
For convenience let us call objects of the form $({\rm tr}_{\mathbb C^{m}}\circ\Psi_t)((\cdot)\otimes\omega)$ \textit{Markov-Stinespring curves}.
Such a construction can be seen as an interpolation between arbitrary dynamics and Stinespring curves (in the sense of Def.~\ref{def_stinespring_curves}): while one may have to add ``many''  environmental degrees of freedom to turn a system of interest into something closed,
it may only take a few extra dimensions
to turn the system into something Markovian,
cf.~also \cite{Breuer04}.
Indeed, proving the above generalization of \cite{Dive15} would imply that the \textit{set} of Markov-Stinespring curves coincides with the set of (finite-dimensional) Stinespring curves.
\section*{Acknowledgments}
I am grateful to
Philippe Faist whose question sparked this paper, to Gunther Dirr for
making me aware of the paper of Dole\v{z}al \cite{Dolezal64}  (cf.~Remark~\ref{rem_cont_ac_app}, Appendix~B), as well as to
Emanuel Malvetti and Fereshte Shahbeigi
for valuable and constructive comments during the preparation of this paper.
Moreover, I would like to thank Sumeet Khatri for proofreading a preliminary version of this manuscript,
as well as the anonymous referee for their helpful comments which led to a substantially improved presentation of the material.
This work has been supported by the Einstein Foundation (Einstein Research Unit on Quantum Devices) and the MATH+ Cluster of Excellence.

\section*{Appendix A: Initial Value Problems and Their Solutions}

Because we are dealing with a generalized notion of differential equations
we should be particularly careful when it comes to the underlying mathematical formalism.
For this we mainly follow \cite[Appendix~C]{Sontag}.
The upshot of this appendix can be found at the end of this section, refer to Lemma~\ref{app_A_prop_1}.
Given any mapping
$f:I\to V$ where 
$I$ is any connected subset of $\mathbb R$ (called \textit{interval}) 
and $V$ is any finite-dimensional normed space
one calls $f$
\begin{itemize}
\item \textit{locally absolutely continuous} if the restriction $f|_K$ to each compact interval $K\subseteq I$ is absolutely continuous, i.e.~for every such $K$ and every $\varepsilon>0$ there exists $\delta>0$ such that
for every finite sequence of non-empty pairwise disjoint subintervals $(a_j,b_j)\subset K$, $j=1,\ldots,k$
which satisfies $\sum_{j=1}^k(b_j-a_j)<\delta$ it holds that
$\sum_{j=1}^k\|f(b_i)-f(a_i)\|<\varepsilon$.
The collection of all such functions is denoted by $\mathsf{AC}_{\sf loc}(I,V)$.
\item \textit{locally integrable} if $f$ is measurable and $\int_K\|f(t)\|\,{\rm d}t<\infty$ for all compact intervals $K\subseteq I$.
The collection of all such functions is denoted by $L^1_{\sf loc}(I,V)$.
\item \textit{locally essentially bounded} if $f$ is measurable and the 
restriction $f|_J$ to any bounded subinterval $J\subseteq I$ is essentially 
bounded, that is, for every such $J$ there exists $K\subseteq V$ compact such 
that $f(t)\in K$ for almost all $t\in J$.
The collection of all such functions is denoted by $L^\infty_{\sf loc}(I,V)$.
\end{itemize}
The reason we consider the local versions of these notions is that
we are also dealing with dynamics on an infinite interval $I$ (e.g., $I=[0,\infty)$)
where one cannot guarantee absolute continuity but 
``only'' its
local counterpart.

The key property of absolutely continuous functions is that they are precisely those functions which can be recovered via their derivative.
More precisely,
an absolutely continuous function $f:[a,b]\to V$ is differentiable almost everywhere, $\dot f$ is integrable, and
$$
f(t)=f(a)+\int_a^t \dot f(\tau)\,{\rm d}\tau\qquad\text{ for all }t\in[a,b]\,,
$$
cf.~\cite[Thm.~7.20]{Rudin86}.
This readily transfers to the local version of absolute continuity
because $\mathbb R$ is a $\sigma$-finite measure space:
If $f:I\to V$ is locally absolutely continuous, then $f$ is differentiable almost everywhere, $\dot f$ is locally integrable, and for all $t_0,t\in I$ one has
$f(t)=f(t_0)+\int_{t_0}^t \dot f(\tau)\,{\rm d}\tau$.
Another property we need---where $I$ and $V$ are chosen as before---is that
if $V$ is a
normed algebra, then the pointwise products $f\cdot g, g\cdot f:I\to V$
of functions $f\in L^1_{\rm loc}(I,V)$, $g\in L^\infty_{\rm loc}(I,V)$ are in
$L^1_{\rm loc}(I,V)$ (due to the norm being submultiplicative).

A note on topologies on these spaces:
$(L^1,\|\cdot\|_1)$, $(L^\infty,\|\cdot\|_\infty)$ are well-known to be Banach 
spaces \cite[Thm.~13.5]{MeiseVogt97en}, and there are many norms such that the same is true for $\mathsf{AC}([a,b],V)$ such as, for example, $f\mapsto\|f\|_\infty+\|\dot f\|_1$, $f\mapsto|f(t_0)|+\|\dot f\|_1$ for any $t_0\in [a,b]$, or $f\mapsto\|f\|_1+\|\dot f\|_1$ (the latter turns $\mathsf{AC}$ into 
the Sobolev space $L^{1}_{1}([a,b])$ \cite[Ch.~1.1.3]{Mazya11}).
However, if $I$ is not compact then their local counterparts are not Banach 
spaces;
instead $L^1_{\rm loc}, L^\infty_{\rm loc}, \mathsf{AC}_{\rm loc}$ can 
be equipped with a metric which turns them into a Fr\'{e}chet space 
\cite[Lemma 5.17~ff.]{MeiseVogt97en}.\medskip

With this we have introduced the necessary language for studying more general initial value problems.
Let $E:I\times\Omega\to V$ with $I\subseteq\mathbb R$ an interval and $\Omega\subseteq V$ open.
If for $t_0\in I$, $f_0\in\Omega$ there exists $
f\in\mathsf{AC}_{\sf loc}(I,\Omega)$ which satisfies
\begin{equation}\label{eq:gen_ivp}
f(t)=f_0+\int_{t_0}^t E(\tau,f(\tau))\,{\rm d}\tau
\end{equation}
for all $t\in I$, then we say $f$ is a \textit{solution} of~\eqref{eq:gen_ivp}.
If this $f$ is differentiable it even satisfies
$
\dot f(t)=E(t,f(t))$ with initial condition $f(t_0)=f_0\,.
$
Under certain assumptions on the map $E$ one can guarantee the existence of unique solutions of~\eqref{eq:gen_ivp} on a maximal subinterval $J\subseteq I$, $t_0\in J$ \cite[Thm.~54]{Sontag}.
However, as our setting is a lot more specific we may
cut some corners and get straight to linear differential equations $\dot f(t)=A(t)f(t)$, $f(t_0)=f_0$---respectively their integral counterpart
\begin{equation}\label{eq:lin_ivp}
f(t)=f_0+\int_{t_0}^t A(\tau)f(\tau)\,{\rm d}\tau
\end{equation}
---where $t_0\in I$, $f_0\in V$, and $A:I\to\mathcal L(V)$ maps to the linear operators on $V$.
The minimal requirement on $A$ needed for this problem to be well-posed is that it is locally integrable (as solutions are defined to be locally absolutely continuous).
It turns out that this is also sufficient in some sense:
for locally integrable $A$
it follows---again from \cite[Thm.~54]{Sontag}---that~\eqref{eq:lin_ivp}
has a unique maximal solution.
Even better, in our setting $f$ takes values in the quantum states---or, for the operator lift, in the unitary matrices or the quantum channels---all of which are bounded sets \cite{PG06}.
Thus \cite[Prop.~C.3.6]{Sontag} guarantees the existence of
a unique solution\footnote{
Strictly speaking the result quoted from Sontag's book requires the solution to live inside a compact space at all times. However, a bounded set in a normed space is one which is contained in $\overline{B_K(0)}$ for some $K>0$, and in finite dimensions
the latter is compact.
}
for all times $t>t_0$.
For convenience we summarize this in the following lemma:
\begin{lemma}\label{app_A_prop_1}
Let $n\in\mathbb N$, $t_f\in(0,\infty]$.
For all $H:[0,t_f)\to i\mathsf u(n)$ locally integrable and all $U_0\in\mathsf U(n)$ the integral equation
\begin{equation}\label{eq:int-eq-U}
U(t)=U_0-i\int_0^t H(\tau)U(\tau)\,{\rm d}\tau
\end{equation}
has a unique solution $U:[0,t_f)\to\mathsf U(n)$ (i.e.~$U$ is locally absolutely continuous and satisfies~\eqref{eq:int-eq-U} for all $t\in I$).
In this case $U$ satisfies the differential equation $
\dot U(t)=-iH(t)U(t)$, $U(0)=U_0$
for almost all\,\footnote{
If $U$ is differentiable at a so-called Lebesgue point $t_0$ of $H$, then $\dot U(t_0)=-iH(t_0)U(t_0)$ \cite[Thm.~7.11]{Rudin86};
moreover $U$ is differentiable almost everywhere \cite[Thm.~7.20]{Rudin86}
and almost every point is a Lebesgue point \cite[Thm.~7.7]{Rudin86}.
Simple counterexamples show that, in general, one needs that $t_0$ is a Lebesgue point of $H$ to get from~\eqref{eq:schrodinger-int} to~\eqref{eq:schrodinger}, cf.~also \cite[Ch.~4, Ex.~6]{GOCounterexamples03}.
}
$t\in [0,t_f)$.
Moreover, given any $\rho_0\in\mathbb D(\mathbb C^n)$ the Liouville-von Neumann equation $$\dot\rho(t)=-i[H(t),\rho(t)]\,,\quad\rho(t_0)=\rho_0$$ (more 
precisely its integral version $\rho(t)=\rho_0-i\int_0^t[H(\tau),\rho(\tau)]\,{\rm d}\tau$)
has a unique solution
 $\rho:[0,t_f)\to\mathbb D(\mathbb C^n)$ given by $\rho(t)=U(t)\rho_0U(t)^*$
 where $U$ is the solution to~\eqref{eq:int-eq-U} for $U_0=\mathbbm1$.
\end{lemma}

\section*{Appendix B: Evolutions of Isometries}

In this appendix we derive a result which extends certain curves of isometries to curves of unitaries and, most importantly, the construction conserves any type of regularity.

\begin{lemma}\label{lemma_app_A}
Let $m,n\in\mathbb N$, $m\geq n$, $t_f\in(0,\infty]$ as well as $V:[0,t_f)\to\mathbb C^{m\times n}$ be given such that $V(t)$ is an isometry for all $t\in[0,t_f)$.
If $V$ is locally absolutely continuous, then there exists $U:[0,t_f)\to\mathsf U(m)$, $U(0)=\mathbbm1$ locally absolutely continuous such that
\begin{equation}\label{eq:lemma_app_A_1}
V(t)=U(t)V(0)\qquad\text{ for all }t\in[0,t_f)\,.
\end{equation}
Moreover, any regularity of $V$ pertains to $U$.
\end{lemma}
\begin{proof}
Let us distinguish two cases here:
1. Assume $V(0)=(e_1\,\ldots\,e_n)$ where $e_j\in\mathbb C^m$, $j=1,\ldots,m$ is the $j$-th standard basis vector.
Then~\eqref{eq:lemma_app_A_1} implies $U(t)=(V(t)\ *)$ for all $t$. Thus if we can find a (time-dependent, locally absolutely continuous) orthonormal basis of $({\sf ran}\,V(t))^\perp$ then we can fill up $U(t)$ with it.
This amounts to diagonalizing the orthogonal projection $t\mapsto\mathbbm1-V(t)V(t)^*$
for which Kato provides a construction \cite[Ch.~II, §4.5]{Kato80}:
First define $Q:[0,t_f)\to\mathbb C^{m\times m}$ via
\begin{align*}
Q(t):=\,&\tfrac12\big[ \tfrac{d}{dt}(V(t)V(t)^*),V(t)V(t)^* \big]+\tfrac12\big[ \tfrac{d}{dt}(\mathbbm1-V(t)V(t)^*),\mathbbm1-V(t)V(t)^* \big]\\
=\,&(\mathbbm1-V(t)V(t)^*)\dot V(t) V(t)^*-V(t)\dot V(t)^*(\mathbbm1-V(t)V(t)^*)
\end{align*}
for all $t\in[0,t_f)$.
Note that $Q$ is in $L^1_{\sf loc}([0,t_f))$ because it is the sum of products of the $L^1_{\sf loc}$ function $\dot V(t)$ and the bounded (hence\footnote{
Because local absolute continuity of some $f:I\to V$ implies measurability of $f$ (for all $t_0\in I$, $\varepsilon>0$ the intersection $\overline{B_\varepsilon(t_0)}\cap I$
is compact; thus by assumption
$f|_{\overline{B_\varepsilon(t_0)}\cap I}$ is absolutely continuous, hence continuous, hence measurable),
locally absolutely continuous plus boundedness guarantees locally essentially bounded.
}
$L^\infty_{\sf loc}$)
functions $V(t),V(t)^*,\mathbbm1-V(t)V(t)^*$ (cf.~Appendix~A).
Moreover $Q(t)$ is skew-Hermitian for all $t\in[0,t_f)$.
Together this shows that the linear differential equation (formally: the respective integral equation) $\dot W(t)=Q(t)W(t)$, $W(0)=\mathbbm1$ has a unique solution $W:[0,t_f)\to\mathsf U(m)$.
Then Kato's result shows that
$W(t)(\mathbbm1-V(0)V(0)^*)W(t)^*=\mathbbm1-V(t)V(t)^*$ for all $t\in[0,t_f)$.
We claim that
$$
U(t):=\begin{pmatrix}
V(t)\ \ W(t)\begin{pmatrix}
0\\\mathbbm1_{m-n}
\end{pmatrix}
\end{pmatrix}
$$
---i.e.~if $W(t)=(A(t)\ B(t))$ where $A(t)$ is of same size as $V(t)$, then $U(t)$ equals $(V(t)\ B(t))$---solves~\eqref{eq:lemma_app_A_1} and has all the desired properties:
\begin{itemize}
\item At time zero
$$
U(0)=\begin{pmatrix}
V(0)&W(0)\begin{pmatrix}
0\\\mathbbm1_{m-n}
\end{pmatrix}
\end{pmatrix}=\begin{pmatrix}
\mathbbm1_n&0\\0&\mathbbm1_{m-n}
\end{pmatrix}=\mathbbm1\,.
$$
\item $U(t)$ is unitary: a straightforward computation shows
$
U(t)U(t)^*=V(t)V(t)^*+W(t)(\mathbbm1-V(0)V(0)^*)W(t)
=V(t)V(t)^*+\mathbbm1-V(t)V(t)^*=\mathbbm1$.
\item $U(t)V_0=U(t)(e_1\,\ldots\,e_n)=V(t)$ for all $t$, meaning~\eqref{eq:lemma_app_A_1} holds.
\item $U(t)$ is locally absolutely continuous because $V$ and $W$ are.
\item If $V$ has some level of regularity (i.e.~$C^k$ or analytic) on some open subinterval of $[0,t_f)$,
then $Q$ is $C^{k-1}$ by definition so $W$ is $C^k$ again. Thus $U$ is $C^k$
on said interval because both $V$ and $W$ are.
\end{itemize}
2. Now for the general case: given
$V$ with $V(0)$ arbitrary,
there certainly
exists $U_0\in\mathsf U(m)$ such that $U_0^*V(0)=(e_1\,\ldots\,e_n)$. Defining
$\tilde V(t):=U_0^*V(t)$ turns~\eqref{eq:lemma_app_A_1} into
$\tilde V(t)=U_0^*V(t)=U_0^*U(t)V(0)=U_0^*U(t)U_0\tilde V(0)
$.
Because $\tilde V(0)=(e_1\,\ldots\,e_n)$ we 
can follow step 1~above
to find $\tilde U(t)$ locally absolutely continuous with $\tilde U(0)=\mathbbm1$ such that 
$\tilde V(t)=\tilde U(t)\tilde V(0)$,.
This
in turn yields a locally absolutely continuous solution of~\eqref{eq:lemma_app_A_1} 
via $ U(t):=U_0\tilde U(t)U_0^*$.
\end{proof}

\begin{remark}\label{rem_cont_ac_app}
Lemma~\ref{lemma_app_A} continues to hold
if local absolute continuity of $V$ is replaced by regular continuity:
as explained in the proof all we need is a (time-dependent, continuous) orthonormal basis of $({\sf ran}\,V(t))^\perp={\sf ker}\,(V(t)V(t)^*)$.
But the latter is equal to ${\sf ran}\,(\mathbbm1-V(t)V(t)^*)$ \cite[5.15]{Rudin91}
so the existence of a continuous basis is guaranteed by an old result of Dole\v{z}al \cite{Dolezal64} (there, set $A(t)=\mathbbm1-V(t)V(t)^*$).
Finally this basis can be made orthonormal via the usual Gram-Schmidt procedure which, importantly, preserves continuity.
\end{remark}

\section*{Appendix C: Auxiliary Results Regarding (Lipschitz) Continuity}
Recall that, given metric spaces $(X,d_X)$, $(Y,d_Y)$, a map $f:X\to Y$ is called 
\textit{Lipschitz continuous} (or just \textit{Lipschitz}) if there exists $K>0$ such that $d_Y(f(x_1),f(x_2))\leq Kd_X(x_1,x_2)$ for all $x_1,x_2\in X$. Any such $K$ is then referred to as \textit{Lipschitz constant} of $f$.
First we need the following well-known relation between the norm of the derivative and Lipschitz constants \cite[(8.5.2)]{Dieudonne69}:
\begin{lemma}\label{lemma_A2}
Let $K\geq 0$, a real interval $I$, a complete normed space $(X,\|\cdot\|_X)$, and a continuous mapping $f:I\to X$ be given.
If there exists an at most countable subset $I_0$ of $I$ such that for all $t\in I\setminus I_0$, $f$ has a derivative $f'(t)\in X$ at $t$ (w.r.t.~$I$)\;\footnote{
This means that there exists $f'(t)\in X$ such that
$
\|\frac{f(t+h)-f(t)}{h}-f'(t)\|_X\to 0
$
as $h\to 0$ under the additional condition $t+h\in I$ for all $h$ small enough \cite[Def.~3.2.3]{HillePhillips}.
}
such that $\|f'(t)\|_X\leq K$, then $f$ is Lipschitz continuous with Lipschitz constant at most $K$.
\end{lemma}
The next auxiliary result is concerned with estimating the distance between Lipschitz curves which intersect at certain points.
\begin{lemma}\label{lemma_A1}
Let $I\subseteq\mathbb R$ be a bounded from below, closed interval. Moreover let $I_f$ be a closed and isolated subset of $I$ with strictly increasing enumeration $(t_j)_{j\in{\sf N}}$, ${\sf N}\subseteq\mathbb N$ 
such that
$$
\Delta(I_f):=\sup
\big\{|s-t|:  s,t\in I\setminus I_f\text{ such that }{\rm conv}(\{s,t\})\subseteq I\setminus I_f  \big\}<\infty\,,
$$
that is, each connected component of $I\setminus I_f$ has (finite) length at most $\Delta(I_f)$.
Now given any normed space $(X,\|\cdot\|_X)$ and $f,g:I\to X$ Lipschitz continuous with respective Lipschitz constants $K_f,K_g>0$,
if one has $f(t)=g(t)$ for all $t\in I_f$, then
$$
\|f-g\|_{\rm sup} \leq (K_f+K_g)\Delta(I_f)\,.
$$
\end{lemma}
\begin{proof}
W.l.o.g.~${\sf N}=\mathbb N$; the case $|{\sf N}|<\infty$ is proven analogously.
Defining $t_0:=\min I$ (which exists by assumption)
we find
$I=[t_0,\infty)=\bigcup_{j=0}^{\infty}[t_j,t_{j+1})$.
This lets us compute
\begin{align*}
\|f-g\|_{\rm sup}&=\sup_{t\in I}\|f(t)-g(t)\|_X\\
&=\max\big\{\sup_{t\in[t_0,t_{1})}\|f(t)-g(t)-f(t_1)+g(t_1)\|_X,\\
&\hphantom{=\sup\big\{}\ \sup_{j\in\mathbb N}\sup_{t\in[t_j,t_{j+1})}\|f(t)-g(t)-f(t_j)+g(t_j)\|_X\big\}\\
&\leq\max\big\{\sup_{t\in[t_0,t_{1})}\|f(t)-f(t_1)\|_X+\sup_{t\in[t_0,t_{1})}\|g(t)-g(t_1)\|_X,\\
&\hphantom{=\sup\big\{}\ \sup_{j\in\mathbb N}\sup_{t\in[t_j,t_{j+1})}\big(\|f(t)-f(t_j)\|_X+\|g(t)-g(t_j)\|_X\big)\big\}\\
&\leq (K_f+K_g) \max\big\{\sup_{t\in[t_0,t_1)}|t-t_1|,\sup_{j\in\mathbb N}\sup_{t\in[t_j,t_{j+1})}|t-t_j|\big\}\\
&=(K_f+K_g)\sup_{j\in\mathbb N_0}|t_j-t_{j+1}|= (K_f+K_g) \Delta(I_f)\,.\qedhere
\end{align*}
\end{proof}

We conclude this appendix by
asserting Lipschitz continuity of the matrix logarithm on the unitary group.
This complements the more known fact that the exponential on the unitary algebra has Lipschitz constant $1$, i.e.~for all $H_1,H_2\in\mathbb C^{n\times n}$ Hermitian it holds that $\|e^{iH_1}-e^{iH_2}\|_\infty\leq\|H_1-H_2\|_\infty$ (follows from Eq.~\eqref{eq:duhamel_U}).

\begin{lemma}\label{lemma_app_C_unitary_gen}
The following statements hold:
\begin{itemize}
\item[(i)]
Given $z\in\mathbb C$, $|z|=1$ there exists $\phi\in(-\pi,\pi]$ such that $z=e^{i\phi}$ and
$
2|\phi|\leq\pi|z-1|
$.
\item[(ii)] Given $U\in\mathbb C^{n\times n}$ unitary there exists $H\in i\mathfrak u(n)$ such that $U=e^{iH}$ and
$$
2\|H\|_\infty\leq\pi\|U-\mathbbm1\|_\infty\,.
$$
\end{itemize}
\end{lemma}
\begin{proof}
(i): Given $z$ on the unit circle there exists (unique) $\phi\in(-\pi,\pi]$ such that $z=e^{i\phi}$. We compute
\begin{align*}
\pi|z-1|=\pi|e^{i\phi}-1|=\pi\sqrt{(\cos(\phi)-1)^2+(\sin(\phi))^2}=\pi\sqrt{2(1-\cos(\phi))}\,.
\end{align*}
Combining the double-angle formula $\cos(\phi)=1-2(\sin(\frac{\phi}{2}))^2$ with Jordan's inequality $\frac{2}{\pi}\leq\frac{\sin(\alpha)}{\alpha}$\,---the latter being true for all $\alpha\in[-\frac{\pi}{2},\frac{\pi}{2}]$ \cite[Ch.~4.3]{Abramowitz72}---one finds 
$1-\cos(\phi)\geq \frac{2\phi^2}{\pi^2}$
for all $\phi\in[-\pi,\pi]$.
This results in the desired estimate
$
\pi|z-1|=\pi\sqrt{2(1-\cos(\phi))}\geq\pi\sqrt{2} \,\frac{\sqrt{2}}{\pi}|\phi|=2|\phi|\,.
$

(ii): Because $U$ is unitary it is normal and has eigenvalues of absolute value one.
Thus there exist $V\in\mathbb C^{n\times n}$ unitary as well as $z_1,\ldots,z_n\in\mathbb C$ with $|z_1|=\ldots=|z_n|=1$ such that $U=V{\rm diag}(z_1,\ldots,z_n)V^*$ \cite[Thm.~2.5.4]{HJ1}.
By (i) there exist $\phi_1,\ldots,\phi_n\in(-\pi,\pi]$ such that
$z_j=e^{i\phi_j}$ and $2|\phi_j|\leq\pi|z_j-1|$ for all $j=1,\ldots,n$.
Define $H:=V{\rm diag}(\phi_1,\ldots,\phi_n)V^*$;
this Hermitian matrix satisfies $U=e^{iH}$ by construction.
Using unitary invariance of the operator norm one finds
\begin{align*}
\pi\|U-\mathbbm1\|_\infty&=\pi\|V{\rm diag}(z_1,\ldots,z_n)V^*-VV^*\|_\infty\\
&=\pi\| {\rm diag}(z_1-1,\ldots,z_n-1) \|_\infty\\
&=\max_{j=1,\ldots,n}\pi|z_j-1|\geq \max_{j=1,\ldots,n}2|\phi_j|=2\|H\|_\infty\,.\qedhere
\end{align*}
\end{proof}

\bibliographystyle{mystyle}
\bibliography{../../../../control21vJan20.bib}
\end{document}